\documentclass[authoryear,12pt,letter,3p]{jowarticle}
\usepackage{graphicx}
\usepackage{amsmath}
\usepackage{amssymb}
\usepackage{setspace}
\usepackage[all]{xy}
\usepackage{amsthm}
\usepackage[shortlabels]{enumitem}
\makeatletter
\renewcommand\section{\@startsection{section}{1}{\z@}{-3.25ex plus -1ex minus -.2ex}{1.5ex plus .2ex}{\normalsize\bf}}
\renewcommand\subsection{\@startsection{subsection}{2}{\z@}{-3.25ex plus -1ex minus -.2ex}{1.5ex plus .2ex}{\normalsize\bf}}
\renewcommand\subsubsection{\@startsection{subsubsection}{3}{\z@}{-3.25ex plus -1ex minus -.2ex}{1.5ex plus .2ex}{\normalsize\bf}}
\makeatother

\newtheorem{thm}{Theorem}

\newtheorem{prop}[thm]{Proposition}

\begin{document}

\begin{frontmatter}
\title{A Brief Comment on Maxwell(/Newton)[-Huygens] Spacetime}
\author{James Owen Weatherall}\ead{weatherj@uci.edu}
\address{Department of Logic and Philosophy of Science\\ University of California, Irvine, CA 92697}
\begin{abstract} I provide an alternative characterization of a ``standard of rotation'' in the context of classical spacetime structure that does not refer to any covariant derivative operator.
\end{abstract}
\end{frontmatter}

\doublespacing

Following recent work by Simon \citet{Saunders} and Eleanor \citet{KnoxNEP}, a flurry of recent papers have addressed the question of how to understand the geometry presupposed by Newtonian gravitational theory, particularly in light of Corollary VI to the Laws of Motion in Newton's \emph{Principia} \citep{WeatherallSaunders, WallaceEFG, WallaceMPNC, TehRecovery, DewarMaxwell}.\footnote{Of course, there is an older literature addressing closely related issues concerning Corollary VI---see, in particular, \citet{SteinNST,SteinPrehistory} and \citet{DiSalle}.  There is also a literature on the closely related question of how to understand the relationship between ``ordinary'' Newtonian gravitation and ``geometrized'' Newtonian gravitation, also known as Newton-Cartan theory: see, for instance, \citet{GlymourTE}, \citet{KnoxFormulation}, and \citet{WeatherallEquivalence}.}  At issue has been the relationship between (1) Saunders' proposal that one can (and should) take the ``correct'' geometry for Newtonian gravitational theory to be that of what \citet{EarmanWEST} called ``Maxwellian spacetime'', and which more recently has been called ``Newton-Huygens spacetime'' \citep{Saunders} or ``Maxwell-Huygens spacetime'' \citep{WeatherallSaunders}, and (2) Knox's proposal that Corollary VI should motivate a move to geometrized Newtonian gravitation (i.e., Newton-Cartan theory).

One (somewhat tangential) thread of this discussion has concerned how to best characterize Maxwellian spacetime, which is supposed to be endowed with spatial and temporal metric structure and with a standard of rotation for smooth vector fields, but which is \emph{not} supposed to pick out a preferred class of inertial trajectories---i.e., Maxwellian spacetime carries something less than a full affine structure.  When \citet{EarmanWEST} introduced Maxwellian spacetime, he defined it using an equivalence class of covariant derivative operators all agreeing on which smooth timelike vector fields are non-rotating;\footnote{Here and in what follows, we consider only torsion-free derivative operators.} \citet{WeatherallSaunders} adopted the same definition.  But one might worry that this approach is problematic.  The worries are that (a) defining Maxwellian spacetime by taking an equivalence class of derivative operators makes reference to structure that one does not attribute to spacetime; and (b) that there is a more direct ``Kleinian'' characterization of the intended structure that, one might think, captures the intrinsic geometry more effectively than introducing more structure than one wants and then equivocating.\footnote{There is an issue, here, which is that alternative approaches all begin with a coordinate system, and then introduce a class of coordinate transformations that leave some structure invariant---a strategy that I understand as introducing extra structure---the coordinate system---and then removing it by taking equivalence classes.  But I will not address this point in what follows.}  More generally, \citet{WallaceCoordinates} has suggested that the example of Maxwellian spacetime, defined using an equivalence class of derivative operators, shows that coordinate-free methods are not an intuitive way of characterizing certain spatiotemporal structures.

I do not want to argue about the relative merits of different ways of characterizing geometry.  But it does seem to me that in the discussions just described, some weight has been placed on a particular presentational choice, originating with Earman but repeated by several others, even though other, perhaps more attractive, choices are available.  In particular, one can characterize a ``standard of rotation'' in just the sense that Earman and others discuss, in a fully covariant, coordinate-free manner, without ever introducing covariant derivative operators and with no equivalence classes in sight.\footnote{One might ask: could one do a similar thing in the case of a nondegenerate metric?  (Or, put more baldly, why is this not a standard notion already?)  The answer is ``yes'', but it is trivial, since every pseudo-Riemannian metric is compatible with a unique torsion-free derivative operator, and so one automatically gets \emph{more} than a standard of rotation from the metric alone.}  This structure permits an alternative characterization of Maxwellian spacetime that avoids the worries mentioned above; it should be of interest irrespective of one's preferences concerning transformation-based and coordinate-free methods, insofar as it provides an intrinsic characterization of the relevant structure.  My purpose in this short note is to show how this works.\footnote{\label{Dewar} Of course, this alternative formulation of Maxwellian spacetime only draws more attention to the question of whether this structure is sufficient to formulate Newtonian gravitational theory.  One would like to find a coordinate-free presentation of the theory that makes use of precisely Maxwellian spacetime, as characterized here, and nothing else---a version, say, of Neil Dewar's ``Maxwell gravitation'' expressed using only a standard of rotation, \citep{DewarMaxwell}.  I do not attempt that here, though see footnote \ref{DewarRevisited} and the surrounding discussion for a first step in that direction.}  Along the way, I make some remarks about spatial geometry in classical spacetime structures that may be of independent interest.

In what follows, let $M$ be a smooth four-manifold.\footnote{We assume all of the manifolds we consider are connected, paracompact, and Hausdorff.}   A \emph{temporal metric} on $M$ is a closed, non-vanishing one-form $t_a$; a \emph{spatial metric} on $M$ is a smooth, symmetric tensor field $h^{ab}$, which admits, at each point, a collection of four vectors $\overset{i}{\sigma}_a$, for $i=0,1,2,3$, such that $h^{ab}\overset{i}{\sigma}_a\overset{j}{\sigma}_b=1$ if $i=j=1,2,3$ and $0$ otherwise.  A temporal metric $t_a$ and spatial metric $h^{ab}$ are \emph{compatible} if $h^{ab}t_a=\mathbf{0}$.\footnote{For a discussion of these notions, including an account of why the term ``metric'' is appropriate in each case, see \citet[\S 4.1]{MalamentGR}.}  In what follows, we will limit attention to spatial metrics that are compatible with some temporal metric (or other).  We will say that a covariant derivative operator $\nabla$ on $M$ is \emph{compatible} with temporal and spatial metrics $t_a$ and $h^{ab}$ if $\nabla_a t_b=\mathbf{0}$ and $\nabla_a h^{bc}=\mathbf{0}$.

Fix a spatial metric $h^{ab}$ on $M$.  We will say that a vector $\xi^a$ at a point $p$ in $M$ is \emph{timelike} if there exists a non-vanishing covector $\tau_a$ such that $h^{ab}\tau_b=0$ and $\xi^a\tau_a\neq 0$; otherwise it is \emph{spacelike}.\footnote{Observe that we have defined our notion of timelike and spacelike in a way that does not refer to a temporal metric.}  It follows immediately that, at any point $p$, the spacelike vectors at $p$ form a three dimensional subspace $S_pM$ of the tangent space at $p$, $T_pM$.  Given a temporal metric $t_a$, a timelike vector $\xi^a$ will be called \emph{unit} if $|\xi^at_a|=1$.

Let us now introduce the following notation.\footnote{This sort of ``mixed index'' notation is a generalization of the abstract index notation; it is described in more detail in, for instance, \citet{WeatherallFBYMGR}; see also \citet{GerochPDE}.}  Instead of using the usual Latin indices, we will write, for spacelike vectors and vector fields, underlined Latin indices, so that a spacelike vector $\xi$ will be written $\xi^{\underline{a}}$.  Likewise, given a linear functional $\lambda$ acting on spacelike vectors, we will write $\lambda_{\underline{a}}$.  Finally, we can consider tensor fields with (some) underlined indices, as in $\lambda^{a\underline{b}}{}_{\underline{c}d}$: in such cases, an underlined index appearing in a contravariant (raised) position indicates that, for any covector $\tau_a$, if $h^{ab}\tau_b=\mathbf{0}$, then $\tau_a$, contracted with that index, yields zero; meanwhile an underlined index appearing in a covariant (lowered) position indicates that the relevant action is restricted to spacelike vectors (i.e., it is not defined for timelike vectors).  Note that we may always freely remove the lines under contravariant indices, since every spacelike vector at a point is in particular a vector at that point; and we may freely add lines under any covariant (lowered) indices, since every linear functional on tangent vectors at a point may be restricted to spacelike vectors at that point.  Hence, we may write $h^{ab}$ as $h^{\underline{a}\underline{b}}$ and, for any temporal metric $t_{a}$, we have $t_{\underline{a}}=\mathbf{0}$.  But we cannot generally add lines under contravariant indices, since not all tangent vectors are spacelike, and we cannot remove them from covariant indices, since linear functionals on spacelike vectors will not have unique extensions to all tangent vectors.  We will call underlined indices \emph{spatial indices}.

Given the structure defined so far, one can make sense of a \emph{spatial derivative operator} $D$ on $M$, which gives a standard for differentiation of smooth fields with (only) spatial indices in spacelike directions.  I make this idea precise below, but the details are not essential for stating the main claim.  The basic fact about spatial derivative operators that matters for what follows, proved in Prop. \ref{uniqueness} below, is that given a spatial metric $h^{ab}$, there exists a unique spatial derivative operator $D$ with the property that $D_{\underline{a}}h^{\underline{b}\underline{c}}=\mathbf{0}$.  Thus the structure already defined determines a unique spatial derivative operator, in much the same way that a pseudo-Riemannian metric determines a unique derivative operator.\footnote{Note the difference from the presentation in \citet[\S4.1]{MalamentGR}: he defines a spatial derivative operator, but does so only relative to (1) a specific temporal metric $t_a$ and (2) a unit timelike vector field $\xi^a$; moreover, the spatial derivative operator he defines acts, in principle, on arbitrary smooth tensor fields on $M$.  There is nothing wrong with this, of course, and I make use of the same construction in the Proof of Prop. \ref{uniqueness}.  But it perhaps obscures the sense in which we get a unique spatial derivative operator from the spatial geometry alone, and given the purpose of the present note, it seems judicious to avoid any appearances of invoking structure beyond what is strictly needed.}

We can now make the central point.  Fix a temporal metric $t_a$ in addition to $h^{ab}$.  A \emph{standard of rotation} compatible with $t_a$ and $h^{ab}$ is a map ${\circlearrowright}$ from pairs $(x,\xi^a)$, where $\xi^a$ is a smooth vector field on $M$ and $x$ is an index distinct from $a$, to smooth, antisymmetric tensor fields ${\circlearrowright}^x\xi^a = {\circlearrowright}^{[x}\xi^{a]}$ on $M$, satisfying the following conditions:
\begin{enumerate}
\item Given any two smooth vector fields $\xi^a$ and $\eta^a$, ${\circlearrowright}^n(\xi^a + \eta^a) = {\circlearrowright}^n\xi^a + {\circlearrowright}^n\eta^a$;
\item Given any smooth vector field $\xi^a$ and any smooth scalar field $\alpha$, ${\circlearrowright}^n(\alpha\xi^a)=\alpha{\circlearrowright}^n\xi^a + \xi^{[a} d^{n]}\alpha$;\footnote{The operator $d$ is the exterior derivative.  Here and throughout, we raise indices on derivative operators with the spatial metric $h^{ab}$, i.e., $d^n\alpha=h^{mn}d_n\alpha$.}
\item Given any smooth vector field $\xi^a$, if $d_a(\xi^n t_n)=\mathbf{0}$, then ${\circlearrowright}^n\xi^a$ is spacelike in both indices;
\item Given any smooth spacelike vector field $\sigma^{\underline{a}}$, ${\circlearrowright}^{\underline{n}}\sigma^{\underline{a}} =D^{[\underline{n}}\sigma^{\underline{a}]}$; and
\item Given any smooth vector field $\xi^a$, $\delta^a{}_m{\circlearrowright}^n\xi^m = {\circlearrowright}^n(\delta^a{}_m\xi^m)$.
\end{enumerate}

We then have the following proposition characterizing standards of rotation.
\begin{prop}\label{rotation}Let $M$ be a smooth, connected, paracompact, Hausdorff four-manifold, and let $t_a$ and $h^{ab}$ be temporal and spatial metrics on $M$, respectively.\footnote{This result can be extended to other dimensions, but it is not clear that the extra generality is of any interest.}  Then the following hold.
 \begin{enumerate}[(1)]
 \item \label{derivtorot} Given any covariant derivative operator $\nabla$ on $M$ compatible with $t_a$ and $h^{ab}$, there exists a unique standard of rotation ${\circlearrowright}$ on $M$, also compatible with $t_a$ and $h^{ab}$, such that for any smooth vector field $\xi^a$, ${\circlearrowright}^n\xi^a = \nabla^{[n}\xi^{a]} (:= h^{m[n}\nabla_m\xi^{a]})$;
 \item \label{rottoderiv} Given any standard of rotation ${\circlearrowright}$ on $M$ compatible with $t_a$ and $h^{ab}$, there exists a covariant derivative operator on $M$, also compatible with $t_a$ and $h^{ab}$, such that for any smooth vector field $\xi^a$, ${\circlearrowright}^n\xi^a = \nabla^{[n}\xi^{a]}$.
 \item \label{flatness} Given a standard of rotation ${\circlearrowright}$ on $M$ compatible with $t_a$ and $h^{ab}$, if (a) $h^{ab}$ is flat, in the sense that $D_{[\underline{c}}D_{\underline{d}]}\sigma^{\underline{a}}=\mathbf{0}$ for all spacelike vector fields $\sigma^{\underline{a}}$; and (b) there exists a unit timelike vector field $\eta^a$ such that (i) ${\circlearrowright}^n\eta^a = \mathbf{0}$ and (ii) $\mathcal{L}_{\eta}h^{ab}=\mathbf{0}$, then $\nabla$ in \ref{rottoderiv} can be chosen to be flat; conversely, if there exists a flat derivative operator $\nabla$ compatible with $t_a$ and $h^{ab}$, then $h^{ab}$ is flat and its induced standard of rotation in the sense of \ref{derivtorot} satisfies (b), at least locally.
 \item \label{nonuniqueness} Two derivative operators $\nabla$ and $\nabla'$, both compatible with $t_a$ and $h^{ab}$, determine the same standard of rotation in the sense of \ref{derivtorot} iff $\nabla=(\nabla',\sigma^{\underline{a}}t_bt_c)$ for some spacelike vector field $\sigma^{\underline{a}}$.  In particular, if the standard of rotation so-determined admits of any non-rotating unit timelike vector fields, $\xi^a$, this condition holds iff $\nabla^{[a}\xi^{b]}=\mathbf{0} \Leftrightarrow \nabla'^{[a}\xi^{b]}=\mathbf{0}$.
 \end{enumerate}
\end{prop}
\begin{proof}
To establish \ref{derivtorot}, it is sufficient to show that if $\nabla$ is compatible with $t_a$ and $h^{ab}$, then ${\circlearrowright}:(\xi^a,x)\mapsto \nabla^{[n}\xi^{a]}$ satisfies the five conditions above.  (Uniqueness is immediate, since this definition determines the action of ${\circlearrowright}$ on all smooth vector fields.)  The first, second, and fifth are immediate consequences of the properties of derivative operators.  For the third, observe that since $\nabla$ is compatible with $t_a$, given any smooth vector field $\xi^a$ satisfying $d_a(\xi^nt_n)=0$, we have $0=\nabla_a(\xi^n t_n)=t_n\nabla_a\xi^n$.  Hence $t_n(h^{m[n}\nabla_m\xi^{a]})=-t_n(h^{m[a}\nabla_m\xi^{n]})=\mathbf{0}$, and so $\nabla^{[n}\xi^{a]}$ is spacelike in both indices.  Finally, for the fourth, let $\eta^a$ be a unit timelike vector field and let $\hat{h}_{ab}$ be the spatial projector relative to $\eta^a$, i.e., the unique symmetric field such that $\hat{h}_{ab}\eta^b=\mathbf{0}$ and $h^{ab}\hat{h}_{bc} = \delta^a{}_c - \eta^at_c$..  (Unit timelike vector fields always exist. For instance, choose any Riemannian metric $g_{ab}$ on $M$ and let $\eta^a=k g^{ab}t_b$, where $k$ is chosen so that $\eta^at_a=1$.)   Now observe that, as established in the proof of Prop. \ref{uniqueness} below, for any smooth spacelike vector field $\sigma^{\underline{a}}$, $D_{\underline{n}}\sigma^{\underline{a}}=\hat{h}_{\underline{n}m}\hat{h}^{\underline{a}}{}_{x}\nabla^{m}\sigma^x$, where $\hat{h}_{ab}$ is the spatial projection field determined by any smooth, unit timelike vector field $\eta^a$.  Hence $h^{\underline{m}[\underline{n}}D_{\underline{m}}\sigma^{\underline{a}]}=h^{\underline{m}[\underline{n}} \hat{h}_{\underline{m}\underline{y}}\hat{h}^{\underline{a}]}{}_{x}\nabla^{\underline{y}}\sigma^x=h^{m[n}\nabla_m\sigma^{a]}$.

To establish \ref{rottoderiv}, we first show that there always exists a covariant derivative operator compatible with $t_a$ and $h^{ab}$.  Again, let $\eta^a$ be a unit timelike vector field.  Now let $\nabla'$ be an arbitrary derivative operator.  (Again, such exist, by paracompactness.)  From here we proceed in two steps.  First, we define a new derivative operator $\tilde{\nabla}=(\nabla',-\eta^a\nabla'_b t_c)$, which satisfies $\tilde{\nabla}_a t_b = \mathbf{0}$.\footnote{Here we use the fact that the action of any (torsion-free) derivative operator on arbitrary fields may be expressed using any other derivative operator and a smooth field $C^a{}_{bc}$, symmetric in $b$ and $c$.  For details, see \citet[\S1.7]{MalamentGR}.}  (Observe that since $t_a$ is closed, $\nabla_b t_c = \nabla_{(b}t_{c)}$, and so $-\eta^a\overline{\nabla}_b t_c$ has the appropriate symmetry properties.)  We now define a third derivative operator $\overline{\nabla}=(\tilde{\nabla},C^{a}_{bc})$, where
\[
C^a{}_{bc} = -\frac{1}{2}\left[h^{ar}\tilde{\nabla}_r h^{mn} - h^{nr}\tilde{\nabla}_r h^{am} - h^{mr}\tilde{\nabla}_r h^{na}\right]\hat{h}_{bm}\hat{h}_{cn} - t_{(b}\hat{h}_{c)n}\eta^r\tilde{\nabla}_r h^{an},\]
which one can show is compatible with $t_a$ and $h^{ab}$.

Finally, let ${\circlearrowright}$ be any standard of rotation, and define $\nabla = (\tilde{\nabla},2h^{an}t_{(b}\kappa_{c)n})$, where $\kappa_{ab}=\hat{h}_{n[a}\hat{h}_{b]m}\left(\nabla^m\eta^n - {\circlearrowright}^m\eta^n\right)$.  This final derivative operator, $\nabla$, is compatible with $t_a$ and $h^{ab}$; we claim that it is also such that for any smooth vector field $\xi^a$, $\nabla^{[n}\xi^{a]}={\circlearrowright}^n\xi^a$.  To see this, first observe that, by construction, ${\circlearrowright}^n\eta^a = \nabla^{[n}\eta^{a]}$.  (Here we use the fact that ${\circlearrowright}^a\xi^a$ is spatial in both indices to ensure that the action of $\hat{h}_{ab}$ on either index is invertible.)  Now fix any smooth vector field $\xi^a$ and observe that it may be written as $\xi^a = \alpha\eta^a + \sigma^a$, where $\sigma^a$ is spacelike and $\alpha$ is some smooth scalar field (possibly vanishing).  Then ${\circlearrowright}^n\xi^a = \alpha{\circlearrowright}^n\eta^a + \eta^{[a}d^{n]}\alpha + D^{\underline{n}}\sigma^{\underline{a}]} = \nabla^{[n}\xi^{a]}$, as desired.

We now establish \ref{flatness}.  Let ${\circlearrowright}$ on $M$ be compatible with $t_a$ and $h^{ab}$, and suppose conditions (a) and (b) are satisfied.  Since $h^{ab}$ is flat, any derivative operator compatible with it is spatially flat, in the sense that $R^{abcd}=h^{bn}h^{cm}h^{do}R^a{}_{bcd}=\mathbf{0}$.  It follows that if $\nabla$ is compatible with $h^{ab}$, and if there exists a unit timelike vector $\eta^a$ such that $\nabla_a\eta^b=\mathbf{0}$, then $\nabla$ is flat \citep[Prop. 4.2.4]{MalamentGR}.  So let $\eta^a$ be a unit timelike vector such that ${\circlearrowright}^n\xi^a=\mathbf{0}$ and $\mathcal{L}_{\eta}h^{ab}=\mathbf{0}$, and let $\tilde{\nabla}$ be a derivative operator compatible with $t_a$ and $h^{ab}$ whose standard of rotation agrees with ${\circlearrowright}$.  Then define $\nabla = (\tilde{\nabla},t_bt_c\eta^n\tilde{\nabla}_n\eta^a)$.  This derivative operator is compatible with $t_a$ and $h^{ab}$; moreover one can confirm that its standard of rotation agrees with $\tilde{\nabla}$ (and hence ${\circlearrowright}$).  It follows from the first of these facts that $\mathcal{L}_{\eta}h^{ab} = \eta^n\nabla_n h^{ab} - h^{an}\nabla_n\eta^b - h^{nb}\nabla_n\eta^a = -2\nabla^{(a}\eta^{b)} = \mathbf{0}$; and from the second, it follows that $\nabla^{[a}\eta^{b]}=\mathbf{0}$.  Thus $\nabla^a\eta^b=\mathbf{0}$.  But we also have, by construction, $\eta^n\nabla_n\eta^a = \mathbf{0}$, and so $\nabla_a\eta^b = \mathbf{0}$.  It follows that $\nabla$ is flat, compatible with the metrics, and determines standard of rotation ${\circlearrowright}$.  For the converse, observe simply that if $\nabla$ is flat and compatible with $t_a$ and $h^{ab}$, then it is spatially flat (so $h^{ab}$ must be flat) and there exists, at least locally, a constant unit timelike vector field $\eta^a$, which automatically satisfies $\nabla^a\eta^b=\mathbf{0}$.

Finally we establish \ref{nonuniqueness}.  Suppose $\nabla$ and $\nabla'$, both compatible with $t_a$ and $h^{ab}$, determine the same standard of rotation.
It follows that there is some antisymmetric tensor field $\kappa_{ab}$, such that $\nabla' = (\nabla,h^{an}t_{(b}\kappa_{c)n})$; and that for any unit timelike vector field $\xi^a$, $\nabla'^{[a}\xi^{b]}=\nabla^{[a}\xi^{b]}-2h^{o[b}h^{a]n}t_{(n}\kappa_{m)o}\xi^m=\nabla^{[a}\xi^{b]} - \kappa^{ab}$, and so $\kappa^{ab}=\mathbf{0}$.  Thus $\kappa_{ab}=t_{[a}\sigma_{b]}$, for some covector $\sigma_b$, and so $C^a{}_{bc}=\sigma^{\underline{a}} t_bt_c$ for some spacelike vector field $\sigma^{\underline{a}}$.  The final clause of \ref{nonuniqueness} follows immediately.
\end{proof}

We may now define \emph{Maxwellian spacetime} (or, \emph{Newton-Huygens spacetime} or \emph{Maxwell-Huygens spacetime}) as follows: it is a structure $(M,t_a,h^{ab},{\circlearrowright})$, where $M$ is a smooth manifold diffeomorphic to $\mathbb{R}^4$; $t_a$ is a temporal metric on $M$; $h^{ab}$ is a spatial metric on $M$; and ${\circlearrowright}$ is a standard of rotation compatible with $t_a$ and $h^{ab}$.  We further suppose that $t_a$ admits an integral $t:M\rightarrow \mathbb{R}$ (i.e., a smooth field $t$ such that $d_a t=t_a$) and is surjective and whose surfaces of constant value are diffeomorphic to $\mathbb{R}^3$;\footnote{The surjectivity of $t$ captures the idea that time goes on indefinitely in both directions; it is a kind of completeness for temporal metrics.} and that $h^{ab}$, restricted to each of these surfaces, is complete.\footnote{By completeness, here, we mean that the Riemannian metric induced on each of these surfaces is complete in the standard sense.}  Note that on this characterization, no equivalence classes are taken, and in particular, there is no need to refer to a derivative operator (or anything else not already definable from the structure mentioned).\footnote{One could equally well begin with a three dimensional affine bundle over $\mathbb{R}$, and then define $h^{ab}$, $t_a$, and ${\circlearrowright}$ precisely as above.}

Before proceeding, a remark is in order about just what a standard of rotation, in the present sense, allows one to do.  It is important to the discussions of Newtonian gravitation described above, and especially in the context of ``vector relationism'' as presented by \citet{Saunders},\footnote{The expression ``vector relationism'' was apparently coined by \citet{WallaceMPNC}, but to describe Saunders' proposal.} that in Maxwellian spacetime one has a well-defined notion of ``relative acceleration'', which is the rate of change along a timelike curve of a spacelike vector field representing the instantaneous relative velocity of two particles.  If one characterizes Maxwellian spacetime using an equivalence class of derivative operators, one can define this rate of change using any of the derivative operators in the equivalence class and then show that the resulting quantity does not depend on the choice.  But it turns out that one can likewise define a notion of the rate of change of a spacelike vector field in a timelike direction using only the structure of Maxwellian spacetime as we have defined it, without appealing to a derivative operator.\footnote{I am grateful to David Malament for raising this issue.}

In particular, fix a standard of rotation ${\circlearrowright}$ compatible with temporal and spatial metrics $t_a$ and $h^{ab}$ on $M$, let $\sigma^{\underline{a}}$ be a spacelike vector field on $M$, and let $\xi^a$ be a unit timelike vector at a point $p$.  We then define $\xi^n\triangle_n\sigma^{\underline{a}}$, the rate of change of $\sigma^a$ at $p$, in the direction of $\xi^a$, by:
\begin{equation}\label{restrictedDerivative}
\xi^n\triangle_n\sigma^{\underline{a}} = \mathcal{L}_{\xi} \sigma^{\underline{a}} + \sigma_n {\circlearrowright}^n \xi^a - \frac{1}{2} \sigma_n \mathcal{L}_{\xi} h^{an}.
\end{equation}
Here $\mathcal{L}_{\xi}$ is the Lie derivative taken with respect to \emph{any} extension of $\xi^a$ off of $p$, and $\sigma_n$ is any covector with the property that $h^{an}\sigma_n = \sigma^a$.\footnote{To get an intuitive handle on this expression, it is useful to think of the rate of change of $\sigma^{\underline{a}}$ in the direction of $\xi^a$ at $p$ as corresponding to the flow of $\sigma^{\underline{a}}$ along a vector field $\xi^a$ at $p$, corrected for the ways in which $\xi^a$ is itself changing, in the direction of $\sigma^{\underline{a}}$, at $p$.  The first term of Eq. \eqref{restrictedDerivative} captures the ``flow'', while the other two terms, corresponding to the rotation and expansion of the vector field $\xi^a$, describe the change in $\xi^a$ in the direction $\sigma^{\underline{a}}$.}   One can then show, by direct computation, that for any derivative operator $\nabla$ whose standard of rotation agrees with ${\circlearrowright}$ in the sense of Prop. \ref{rotation}, we have $\xi^n\nabla_n\sigma^{\underline{a}}=\xi^n\triangle_n\sigma^{\underline{a}}$.  In this sense, then, one can recover ``relative acceleration'' without introducing any structure beyond a standard of rotation.

I have intentionally written Eq. \eqref{restrictedDerivative} in a suggestive way.  Indeed, one could take this equation to define a new operator $\triangle$, which would be a kind of ``restricted derivative operator'' acting only on tensor fields with spatial indices.  Such an operator might then be used to develop a dynamical theory of spacelike vector fields.\footnote{\label{DewarRevisited} Recall footnote \ref{Dewar}.}  Alternatively, one could provide an abstract definition of $\triangle$, strongly analogous to the definition of the standard of rotation above, and take that operator as a primitive when defining Maxwellian spacetime. Then Eq. \eqref{restrictedDerivative} could be used to define a standard of rotation from a restricted derivative operator, from which one could recover a version of Prop. \ref{rotation} for restricted derivative operators.  This approach, too, would avoid any reference to equivalence classes of covariant derivative operators.

We now turn to making the idea of a spatial derivative operator precise.  A (torsion-free) \emph{spatial derivative operator} on $M$ is a map $D$ from pairs $(\underline{x},\alpha^{\underline{a}_1\cdots\underline{a}_n}_{\underline{b}_1\cdots\underline{b}_m})$, where $\underline{x}$ is an index distinct from all of $\underline{a}_1,\ldots,\underline{a}_n,\underline{b}_1,\ldots,\underline{b}_m$ and $\alpha^{\underline{a}_1\cdots\underline{a}_n}_{\underline{b}_1\cdots\underline{b}_m}$ is a smooth tensor field on $M$ with only underlined indices, to smooth tensors $D_{\underline{x}}\alpha^{\underline{a}_1\cdots\underline{a}_n}_{\underline{b}_1\cdots\underline{b}_m}$,\footnote{Observe that we are requiring that $D$ preserves the underlined character of all indices on $\alpha$.} satisfying the following conditions \citep[cf.][\S 1.7]{MalamentGR}:
\begin{enumerate}
\item $D$ commutes with addition of smooth tensor fields;
\item $D$ satisfies the Leibniz rule with respect to outer multiplication;
\item $D$ commutes with index substitution and with index contraction;
\item For all smooth scalar fields $\alpha$, $D_{\underline{a}}\alpha=d_{\underline{a}}\alpha$; and
\item If $\underline{x}$ and $\underline{y}$ are distinct, then for all smooth scalar fields $D_{\underline{x}}D_{\underline{y}}\alpha=D_{\underline{y}}D_{\underline{x}}\alpha$
\end{enumerate}

We then get the following result.
\begin{prop}\label{uniqueness}
Let $M$ be a smooth, connected, paracompact, Hausdorff four-manifold, and let $h^{ab}$ be a spatial metric on $M$ compatible with some temporal metric.  Then there exists a unique spatial derivative operator $D$ on $M$ such that $D_{\underline{a}}h^{\underline{b}\underline{c}}=\mathbf{0}$.
\end{prop}
\begin{proof}First we establish existence.  Fix some temporal metric $\tau_a$ compatible with $h^{ab}$.  Let $\eta^a$ be a smooth vector field on $M$ everywhere satisfying $\eta^a\tau_a=1$.  (As noted in the proof of Prop. \ref{rotation}, such fields always exist.)  Now, choose any derivative operator $\nabla$ on $M$ such that $\nabla^a h^{bc}=\mathbf{0}$.  (Again, such exist: choose an arbitrary derivative operator $\tilde{\nabla}$ and let $\nabla=(\tilde{\nabla}, C^a{}_{bc})$, where $C^a{}_{bc}=-\frac{1}{2}\left(h^{ar}\tilde{\nabla}_r h^{mn} -h^{nr}\tilde{\nabla}_r h^{am} - h^{mr}\tilde{\nabla}_r h^{na}\right)\hat{h}_{bm}\hat{h}_{cn}$.  Here, as in the proof of Prop. \ref{rotation}, $\hat{h}_{ab}$ is the spatial projector determined by $\eta^a$.)  Finally, define $D$ such that its action, on any scalar field, $\alpha$, is given by $D_{\underline{a}}\alpha=d_{\underline{a}}\alpha$; its action on any spacelike vector field $\xi^{\underline{a}}$, is given by $D_{\underline{n}}\xi^{\underline{a}}=\hat{h}_{\underline{n}m}\hat{h}^{\underline{a}}{}_{x}\nabla^{m}\xi^x$; its action on any spatial covector field $\lambda_{\underline{a}}$ is $\hat{h}_{\underline{n}m}\hat{h}^{\underline{a}}{}_{x}\nabla^{m}h^{x\underline{y}}\lambda_{\underline{y}}$; and its action on arbitrary spatial tensor fields is determined similarly.  Then $D$ inherits from $\nabla$ all of the properties necessary to be a spatial derivative operator; and $D_{\underline{a}}h^{\underline{b}\underline{c}}=\hat{h}_{\underline{n}m}\hat{h}^{\underline{a}}{}_{x}\hat{h}^{\underline{b}}{}_{y}\nabla^{m}h^{xy}=\mathbf{0}$.

Now we establish uniqueness.  First, I claim that given any two spatial derivative operators $D$ and $D'$, there exists a (unique) smooth spatial tensor field $C^{\underline{a}}{}_{\underline{b}\underline{c}}$, symmetric in $\underline{b}$ and $\underline{c}$, such that for any smooth spatial vector field $\xi^{\underline{a}}$, $(D_{\underline{n}} - D'_{\underline{n}})\xi^{\underline{a}} = -C^{\underline{a}}{}_{\underline{n}\underline{m}}\xi^{\underline{m}}$, and for any smooth spatial covector field $\lambda_{\underline{a}}$, $(D_{\underline{n}} - D'_{\underline{n}})\lambda_{\underline{a}} = C^{\underline{m}}{}_{\underline{n}\underline{a}}\lambda_{\underline{m}}$.  (The argument for this follows standard arguments for similar results very closely, and is suppressed; see for instance \citep[Prop. 1.7.3]{MalamentGR}.)  Now suppose that $D_{\underline{a}} h^{\underline{b}\underline{c}} = D'_{\underline{a}}h^{\underline{b}\underline{c}}=\mathbf{0}$.  It follows that $D_{\underline{n}}\hat{h}_{\underline{a}\underline{b}}=D'_{\underline{n}}\hat{h}_{\underline{a}\underline{b}}=\mathbf{0}$, since $\mathbf{0} = D_{\underline{n}}\delta^{\underline{a}}{}_{\underline{b}}=D_{\underline{n}}(h^{\underline{a}\underline{m}}\hat{h}_{\underline{m}\underline{b}})= h^{\underline{a}\underline{m}}D_{\underline{n}}\hat{h}_{\underline{m}\underline{b}}$.  (Here $\delta^{\underline{a}}{}_{\underline{b}}$ is the index substitution operator for spacelike vectors; that $h^{\underline{a}\underline{b}}\hat{h}_{\underline{b}\underline{c}}=\delta^{\underline{a}}{}_{\underline{c}}$ follows from the definition of $\hat{h}_{ab}$ and the fact that $t_{\underline{a}}=\mathbf{0}$.)  But since $h^{\underline{a}\underline{m}}$ is invertible on spatial vectors, this can vanish only if $D_{\underline{n}}\hat{h}_{\underline{m}\underline{b}}=\mathbf{0}$. Likewise for $D'$.  (Observe that this holds for any field $\hat{h}_{\underline{a}\underline{b}}$ determined by a timelike vector field $\eta^a$ as above.)  So we have $(D_{\underline{n}} - D'_{\underline{n}})\hat{h}_{\underline{a}\underline{b}} = C^{\underline{m}}{}_{\underline{n}\underline{a}}\hat{h}_{\underline{m}\underline{b}}+C^{\underline{m}}{}_{\underline{n}\underline{b}}\hat{h}_{\underline{m}\underline{a}}=\mathbf{0}$. But then it also holds that $C^{\underline{m}}{}_{\underline{b}\underline{n}}\hat{h}_{\underline{m}\underline{a}}+C^{\underline{m}}{}_{\underline{b}\underline{a}}\hat{h}_{\underline{m}\underline{n}}=\mathbf{0}$ and $ C^{\underline{m}}{}_{\underline{a}\underline{b}}\hat{h}_{\underline{m}\underline{n}}+C^{\underline{m}}{}_{\underline{a}\underline{n}}\hat{h}_{\underline{m}\underline{b}}=\mathbf{0}$. Now subtracting the second two equations from the first yields that $C^{\underline{m}}{}_{\underline{b}\underline{a}}\hat{h}_{\underline{m}\underline{n}}=\mathbf{0}$, which can hold only if $C^{\underline{a}}{}_{\underline{b}\underline{c}}=\mathbf{0}$.  It follows that $D$ and $D'$ agree on all smooth spatial tensor fields.\end{proof}

\section*{Acknowledgments}
This material is based upon work supported by the National Science Foundation under Grant No. 1331126.  I am grateful to David Malament and Sarita Rosenstock for helpful conversations related to this paper and for comments on a previous version; and to Eleanor Knox for prompting me to write it.  I apologize to the audience at the meeting The Philosophy of Howard Stein and Its Contemporary Influence, held at the University of Chicago, for blurting this out during the question period following Eleanor's talk.

\singlespacing


\begin{thebibliography}{18}
\expandafter\ifx\csname natexlab\endcsname\relax\def\natexlab#1{#1}\fi
\expandafter\ifx\csname url\endcsname\relax
  \def\url#1{\texttt{#1}}\fi
\expandafter\ifx\csname urlprefix\endcsname\relax\def\urlprefix{URL }\fi

\bibitem[{Dewar(2017)}]{DewarMaxwell}
Dewar, N., 2017. Maxwell gravitation, forthcoming in Philosophy of Science.
  Pre-print available at http://philsci-archive.pitt.edu/12470/.

\bibitem[{DiSalle(2008)}]{DiSalle}
DiSalle, R., 2008. Understanding Space-Time. Cambridge University Press, New
  York.

\bibitem[{Earman(1989)}]{EarmanWEST}
Earman, J., 1989. World Enough and Space-Time. The MIT Press, Cambridge, MA.

\bibitem[{Geroch(1996)}]{GerochPDE}
Geroch, R., 1996. Partial differential equations of physics. In: Hall, G.~S.,
  Pulham, J.~R. (Eds.), General Relativity: Proceedings of the Forty Sixth
  Scottish Universities Summer {School} in Physics. SUSSP Publications,
  Edinburgh, pp. 19--60.

\bibitem[{Glymour(1980)}]{GlymourTE}
Glymour, C., 1980. Theory and Evidence. Princeton University Press, Princeton,
  NJ.

\bibitem[{Knox(2011)}]{KnoxFormulation}
Knox, E., 2011. Newton-{C}artan theory and teleparallel gravity: The force of a
  formulation. Studies in History and Philosophy of Modern Physics 42~(4),
  264--275.

\bibitem[{Knox(2014)}]{KnoxNEP}
Knox, E., 2014. Newtonian spacetime structure in light of the equivalence
  principle. The British Journal for the Philosophy of Science 65~(4),
  863--888.

\bibitem[{Malament(2012)}]{MalamentGR}
Malament, D.~B., 2012. Topics in the Foundations of General Relativity and
  Newtonian Gravitation Theory. University of Chicago Press, Chicago.

\bibitem[{Saunders(2013)}]{Saunders}
Saunders, S., 2013. Rethinking {N}ewton's \emph{Principia}. Philosophy of
  Science 80~(1), 22--48.

\bibitem[{Stein(1967)}]{SteinNST}
Stein, H., 1967. Newtonian space-time. The Texas Quarterly 10, 174--200.

\bibitem[{Stein(1977)}]{SteinPrehistory}
Stein, H., 1977. Some philosophical prehistory of general relativity. In:
  Earman, J., Glymour, C., Stachel, J. (Eds.), Foundations of Space-Time
  Theories. University of Minnesota Press, Minneapolis, MN, pp. 3--49.

\bibitem[{Teh(2017)}]{TehRecovery}
Teh, N., 2017. Recovering recovery: On the relationship between gauge symmetry
  and trautman recovery, forthcoming in Philosophy of Science.

\bibitem[{Wallace(2016{\natexlab{a}})}]{WallaceEFG}
Wallace, D., 2016{\natexlab{a}}. Fundamental and emergent geometry in newtonian
  physics, http://philsci-archive.pitt.edu/12497/.

\bibitem[{Wallace(2016{\natexlab{b}})}]{WallaceCoordinates}
Wallace, D., 2016{\natexlab{b}}. Who's afraid of coordinate systems? an essay
  on representation of spacetime structure, forthcoming in Studies in History
  and Philosophy of Modern Physics. Preprint available at:
  http://philsci-archive.pitt.edu/11988/.

\bibitem[{Wallace(2017)}]{WallaceMPNC}
Wallace, D., 2017. More problems for newtonian cosmology. Studies in History
  and Philosophy of Modern Physics 57, 35--40.

\bibitem[{Weatherall(2016{\natexlab{a}})}]{WeatherallEquivalence}
Weatherall, J.~O., 2016{\natexlab{a}}. Are {N}ewtonian gravitation and
  geometrized {N}ewtonian gravitation theoretically equivalent? Erkenntnis
  81~(5), 1073--1091.

\bibitem[{Weatherall(2016{\natexlab{b}})}]{WeatherallFBYMGR}
Weatherall, J.~O., 2016{\natexlab{b}}. Fiber bundles, {Y}ang-{M}ills theory,
  and general relativity. Synthese 193~(8), 2389--2425.

\bibitem[{Weatherall(2016{\natexlab{c}})}]{WeatherallSaunders}
Weatherall, J.~O., 2016{\natexlab{c}}. Maxwell-{H}uygens, {N}ewton-{C}artan,
  and {S}aunders-{K}nox spacetimes. Philosophy of Science 83~(1), 82--92.

\end{thebibliography}
\end{document}